%% file: ms.tex
\documentclass[11pt]{article}

\usepackage[T1]{fontenc}
\usepackage[utf8]{inputenc}

\usepackage[margin = 0.9in]{geometry}
\usepackage{booktabs} 
\usepackage{enumitem}

\usepackage{latexsym,amsmath,amsfonts,amssymb,stmaryrd,mathtools}
\usepackage{amsthm}
\usepackage{xcolor}
\usepackage{algorithm, algorithmicx}
\usepackage[noend]{algpseudocode}
\usepackage{graphicx}
\usepackage{hyperref}
\usepackage{wrapfig}
\usepackage{listings, fancyvrb, multicol}
\usepackage{multirow}
\usepackage{times}
\usepackage{comment}
\usepackage{authblk}
\usepackage{xspace}
\usepackage{pifont}
\usepackage{parcolumns}

\include{macros}

\usepackage{environ}

\newif\ifappendix
\NewEnviron{maybeappendix}[1]
{\ifappendix
	\expandafter\global\expandafter\let\csname putmaybeappendix#1\endcsname\BODY%
	\else
	\expandafter\newcommand\csname putmaybeappendix#1\endcsname{}\BODY%
	\fi}
\newcommand{\putmaybeappendix}[1]{\csname putmaybeappendix#1\endcsname}

\appendixtrue

\begin{document}
	\title{Delay-Free Concurrency on Faulty Persistent Memory}
	
	\author[1]{Naama Ben-David}
	\author[1]{Guy E. Blelloch}
	\author[2]{Michal Friedman}
	\author[1]{Yuanhao Wei}
	\affil[1]{Carnegie Mellon University, USA}
	\affil[2]{Technion, Israel}
	\affil[ ]{\textit{\{nbendavi, guyb, yuanhao1\}@cs.cmu.edu}}
	\affil[ ]{\textit{michal.f@cs.technion.ac.il}}
	
	\maketitle	

	\input{abstract}

	\pagebreak

\clearpage

	\input{intro}

\input{model}

	\input{capsules}

	\input{delay}
	\input{cas}

	\input{general}

	\input{optimizing}
	\input{normalized}

\input{rwcas}
	\input{practical}
	\input{experiments}

	\input{figures}
	
	
\clearpage
	
	\bibliographystyle{plainurl}
	\bibliography{biblio, ref}
	
	\newpage
		\appendix
	\input{appendix}
\end{document}

%% file: macros.tex
\definecolor{ao}{rgb}{0.0, 0.5, 0.0}

\newtheorem{theorem}{Theorem}[section]
\newtheorem{lemma}[theorem]{Lemma}

\newtheorem{definition}[theorem]{Definition}

\newcommand{\Hao}[1]{\noindent
	{\color{ao}{\textbf{Hao: }#1}}}

\newcommand{\cascam}{recoverable CAS}
\newcommand{\alg}{\mathcal{A}}
\newcommand{\ppm}{shared cache}
\newcommand{\readop}{\text{Read}}
\newcommand{\cas}{\text{CAS}}

\newcommand{\writeop}{\textit{write}}
\newcommand{\recover}{\textit{Recover}}
\newcommand{\notify}{\textit{notify}}
\newcommand{\var}[1]{\textit{#1}}

\newcommand{\generalSim}{Constant-Delay Simulator}
\newcommand{\generalSimOpt}{Low-Computation-Delay Simulator}
\newcommand{\normalOne}{Persistent Normalized Simulator}

\newcommand{\hcas}[2]{CAS$_{#2}^{#1}$}
\newcommand{\lcas}[1]{cas$_{#1}$}
\newcommand{\hwrite}[2]{WRITE$_{#2}^{#1}$}
\newcommand{\lwrite}[1]{write$_{#1}$}
\newcommand{\hread}[2]{READ$_{#2}^{#1}$}
\newcommand{\lread}[1]{read$_{#1}$}
\newcommand{\hrc}[2]{EXEC$_{#2}^{#1}$}
\newcommand{\lrc}[1]{exec$_{#1}$}
\newcommand{\getobj}{getObjectIndex}
\newcommand{\retire}{retire}

\newcommand{\hide}[1]{}

\makeatletter
\lst@Key{countblanklines}{true}[t]%
{\lstKV@SetIf{#1}\lst@ifcountblanklines}

\lst@AddToHook{OnEmptyLine}{%
	\lst@ifnumberblanklines\else%
	\lst@ifcountblanklines\else%
	\advance\c@lstnumber-\@ne\relax%
	\fi%
	\fi}
\makeatother


%% file: abstract.tex
\begin{abstract}
  Non-volatile memory (NVM) promises persistent main memory that
  remains correct despite loss of power. This has sparked a line of
  research into algorithms that can recover from a system crash.
  Since caches are expected to remain volatile, concurrent data
  structures and algorithms must be redesigned to guarantee that they
  are left in a consistent state after a system crash, and that the
  execution can be continued upon recovery.  However, the prospect of
  redesigning every concurrent data structure or algorithm before it
  can be used in NVM architectures is daunting.

  In this paper, we present a construction that takes any concurrent
  program with reads, writes and CASs to shared memory and makes it
  persistent, i.e., can be continued after one or more processes fault
  and have to restart.  Importantly the converted algorithm has
  constant computational delay (preserves instruction counts on each
  process within a constant factor), as well as constant recovery delay
  (a process can recover from a fault in a constant number of
  instructions).  We show this first for a simple transformation, and
  then present optimizations to make it more practical, allowing for a
  tradeoff for better constant factors in computational delay, for
  sometimes increased recovery delay.  We also provide an optimized
  transformation that works for any normalized lock-free data
  structure, thus allowing more efficient constructions for a large
  class of concurrent algorithms.
		
  Finally, we experimentally evaluate transformations by
  applying them to a queue. We compare the performance of our
  transformations to that of a persistent transactional memory
  framework, Romulus, and to a hand-tuned persistent queue. We show
  that our transformations perform favorably when compared to
  Romulus. Furthermore, while the hand-tuned version sometimes
  outperforms our transformations, the difference is not an
  unreasonable price to pay for the generality and ease of use that we
  provide.
\end{abstract}

%% file: intro.tex
\section{Introduction}

A new wave of memory technology, known as Non-Volatile Memory (NVM), is making its way into modern architectures. NVM is expected to replace DRAM for main memory, and promises many attractive features, including persistence under transient failures (e.g., a power failure). This persistence
introduces the possibility of recovering the data structure from main memory after a system failure, sparking a flurry of research in this area. However, it also introduces potential for inconsistencies, since caches are expected to remain volatile in these new architectures, losing their contents upon a crash. 
There has been a lot of work on developing algorithms for index trees~\cite{chen2015persistent,lee2017wort,venkataraman2011consistent, lejsek2009nv,oukid2016fptree}, lock based data structures \cite{nawab2017dali, chakrabarti2014atlas}, and lock-free data structures \cite{friedman2018persistent, logfree,cohen2018inherent}. 

A natural question that arises is whether we can find general mechanisms that would port algorithms for current machines over to the new persistent setting.
One approach has been the development of persistent transactional memory frameworks~\cite{kolli2016high,liu2017dudetm,memaripour2017atomic,correia2018romulus}. 
This can be an effective approach, although it does not handle code between transactions.
Another approach to achieve arbitrary persistent data structures is the design of persistent universal constructions~\cite{izraelevitz2016linearizability,cohen2018inherent}. In particular, Cohen \textit{et al.}~\cite{cohen2018inherent} present a universal construction that only requires one flush per operation, thereby achieving optimality in terms of flushes. However, universal constructions often suffer from poor performance, because they sequentialize accesses to the data structure. Furthermore, universal constructions are only applicable to data structures with clearly defined operations, and cannot apply to a program as a whole. For these reasons, even a seemingly efficient universal construction leaves more to be desired.

In this paper, we consider simulators that take any concurrent program and transform it by replacing each instruction of the original program with a simulation that has the same effect. We define \emph{persistent simulations}, which exhibit a tradeoff between their \emph{computation delay}, meaning the overhead introduced by the simulation in a run without crashes, and their \emph{recovery delay}, which is the maximum time the simulation takes to recover from a system crash. Our first result is the presentation of a general persistent simulator, called the \emph{\generalSim}, that takes any concurrent program using Reads, Writes, and compare-and-swap (CAS) operations, and simulates it with constant computation and constant recovery delays.

\begin{theorem}
	Any concurrent program that uses Reads, Writes, and CAS operations can be simulated in the persistent memory model with constant computation delay and constant recovery delay.
\end{theorem}

We assume the Parallel Persistent Memory (PPM)
model~\cite{blelloch2018parallel,attiya2018nesting}. The model
consists of P processors, each with a fast local ephemeral memory of
limited size, and sharing a large persistent memory. The model allows
for each processor to fault and restart (independently or
together). On faulting all the processor's state and local ephemeral
memory are lost, but the persistent memory remains. On restart, each
processor has a location in persistent memory that points to a context
from which it loads its registers, including program counter, and
restarts. Results in this model also apply if all processors share the
ephemeral memory as a shared
cache~\cite{cohen2018inherent,cohen2017efficient,friedman2018persistent}---although
in this case it only makes sense if all processors fail together.
Throughout the paper we will use \emph{private cache model} to refer
to the PPM model, and \emph{shared cache model} for the shared cache
variant.

The \generalSim{} is achieved by using a technique called \emph{capsules} \cite{blelloch2018parallel}, which in effect introduces checkpoints on a per processor basis from which the program continues after recovering from a crash. In general, the more capsules there are in a program, the smaller the size of each capsule, and therefore, recovery time decreases. The idea behind the \generalSim{} is to show that we can have constant sized capsules, and that each capsule can be implemented with constant overhead.


In practice, crashes are relatively rare. It is thus important to minimize the computation delay introduced by a persistent simulator, even at the cost of increased recovery time. In the rest of the paper, we therefore present optimizations that can be applied to the simulator to decrease the computation delay.
The first such optimization is just as general as the \generalSim{} in that it applies to any concurrent program. The difference is that we use fewer capsules; we show where boundaries between capsules can be removed, creating capsules that are larger (not necessarily constant sized), to arrive at a smaller computation delay but a larger recovery delay.
The second optimization applies to a large class of lock-free data structures called \emph{normalized data structures}~\cite{timnat2014practical}. In this setting, we show how to further reduce the number of capsules and thus the computation delay.
The idea behind capsules is that the program is broken up into contiguous chunks of code, called \emph{capsules}. Between every pair of capsules, information about the state of the execution is persisted. When a crash occurs in the middle of a capsule, recovery reloads information to continue the execution from the beginning of the capsule. This means that some instructions may be repeated several times. Blelloch \textit{et al.}~\cite{blelloch2018parallel} noted the need for \emph{idempotent} capsules, that is, capsules that are safe to repeat, and gave sufficient conditions to ensure that sequential code is idempotent. However, repetitions of concurrent code can be even more hazardous, as other processes may observe changes that should not have happened. We therefore formalize what it means for a capsule to be correct in a concurrent setting, and show how to build correct concurrent capsules. 

To build correct capsules for general concurrent programs, we must be able to determine whether a modification of a shared variable can safely be repeated. This can be problematic, because often, the execution of a process depends on the return value of its accesses to shared memory. A bad situation can occur if a process has already made a persistent change in shared memory, but crashed before it could persist its operation's return value~\cite{blelloch2018parallel}. Attiya \textit{et al.}~\cite{attiya2018nesting} consider this problem and define \emph{nesting-safe recoverable linearizability (NRL)}, a correctness condition for persistent objects that allows them to be safely nested. In particular, Attiya \textit{et al.}~\cite{attiya2018nesting} introduced the \emph{recoverable CAS}, a primitive which ensures that if a compare-and-swap by process $p$ has successfully changed its target, this fact will be made known to $p$ even if a crash occurs.  

At a high-level we show how to combine the ideas of capsules with a recoverable CAS in a careful way to achieve our results for programs with shared reads and CASes.  To make the simulation work with writes requires some additional ideas.
We use a modification of this recoverable CAS primitive in our capsules to ensure that a program can know whether it should repeat a CAS. We show that the recoverable CAS algorithm satisfies a stronger property that NRL, allowing the recovery to be called even if a crash occurs after the operation has terminated. This property is very important for use in capsules, because when a crash occurs, all operations of the capsule must be recovered, rather than just the most recent one.  


We test our simulations by applying them to the lock-free queue of Michael and Scott~\cite{michael1996simple}, and comparing their performance with two other state-of-the-art implementations: one using the transactional memory framework Romulus~\cite{correia2018romulus}, and the other a hand-tuned detectable queue, known as the LogQueue~\cite{friedman2018persistent}. Because of their generality, we do not expect general constructions to match the performance of specialized implementations. 
Indeed, the LogQueue sometimes outperforms our transformations, but only by about a factor of $1.19$x on $8$ threads; our most optimized transformation even outperforms the LogQueue on lower thread counts. In comparison, the original MichaelScott queue is between $3.3$x to $1.7$x faster than the LogQueue, showing that the inevitable cost of persistence outweighs the extra cost paid for generality by our transformations.
We further show that our simulations, even the least optimized ones, outperform Romulus for the queue construction.

In summary, the contributions of our paper are as follows.
\begin{itemize}
	\item We define persistent simulations, which consider computation and recovery delay, thereby providing a measure of how faithful a simulation is to the original program, and how fast it can recover from crashes.
	\item We present a constant-computation and constant-recovery delay simulation that applies to any concurrent program.
	\item We show optimized simulations that trade-off computation delay for recovery delay, both for general programs and for normalized data structures.
	\item We show that our transformations are practical by comparing them experimentally to state-of-the-art persistent algorithms.
\end{itemize}

%% file: model.tex
\section{Preliminaries}\label{sec:model}

\subsection{Model}
We use the Parallel Persistent Memory (PMM)
model~\cite{blelloch2018parallel,attiya2018nesting}.  It consists of a
system of $n$ asynchronous processes $p_1 \ldots p_n$. Each process
has access to an unbounded \emph{persistent} shared memory that it may
access with Read, Write, and compare-and-swap (CAS) instructions, as
well as a smaller private \emph{volatile} memory, which can be
accessed with standard RAM instructions.  A process can \emph{persist}
the contents of its volatile memory by writing them into a persistent
memory location.  The volatile memory is explicitly managed; it does
not behave like a cache, in that no automatic evictions occur.

Each process may \emph{crash} at any time. Upon a crash, the contents
of a process's private volatile memory is lost, but the persistent
memory remains unchanged.  After a process crashes, it
\emph{restarts}.  On restart, the volatile memory can be in an
arbitrary state, but the persistent memory is in the same state as
immediately before the crash.  To allow for a consistent restart, each
processor has a fixed memory location in the persistent memory
referred to as its \emph{restart pointer location}.  This location
points to a context from which to restart (i.e., a program counter and
some constant number of register values).  On restart this is reloaded
into the registers, much like a context switch.  Processors can
checkpoint by updating the restart pointer.  Furthermore a process can
know whether it has just crashed by calling a special
\texttt{crashed()} function that returns a boolean flag, and resets
once it is called.  We call programs that are run on a processor of
such a machine \emph{persistent programs}.

A similar model for non-volatile memories allows a shared cache and
automatic cache
evications~\cite{cohen2018inherent,cohen2017efficient,friedman2018persistent}.
In this model only the whole system can fail.  In
Section~\ref{sec:practical}, we discuss the difference between the
models in more detail, and point out that all our results are also
valid on the shared cache variant.  Throughout the paper we will use
\emph{private cache model} to refer to the PPM model, and \emph{shared
  cache model} for the shared cache variant.

\subsection{Definitions}
For algorithms that implement concurrent objects, we define executions as follows.
An \emph{execution}, $E$, involves three kinds of \emph{events} for each process $p_i$ in the system; \emph{invocation} events $I_i(op, obj)$, which invoke operation $op$ on object $obj$, \emph{response} events $R_i(op, obj)$, in which object $obj$ responds to $p_i$'s operation, and \emph{crash} events $C_i$. Crash events are not operation- or object-specific. On a crash event, $p_i$ loses all information stored in its volatile variables (but all shared objects remain unaffected). 
A process $p$ takes \emph{steps} in an execution, which constitute atomic accesses to base objects, and together make an \emph{implementation} of the \emph{high-level operations} represented by the invocation and response events of the execution. In this paper, we sometimes refer to steps as \emph{(low-level) instructions}.

A linearizable data structure is said to be \emph{durably linearizable} if at any time, the state of the data structure in persistent memory is consistent with a linearization of the execution up to that point~\cite{izraelevitz2016linearizability}. This means that regardless of when a crash happens, the state in memory remains consistent. A data structure is said to be \emph{detectable}~\cite{friedman2018persistent} if the operations and their return values are persisted, thus allowing a process to recover and continue its execution after a crash.

%% file: capsules.tex
\subsection{Capsules}
\label{sec:definition}
Our goal is to create concurrent algorithms that are persistent and can recover their execution after a crash. The main idea in achieving this is to periodically persist \emph{checkpoints}, which record the state of the execution at the time they are persisted, and from which we can continue our execution after a crash. We call the code between any two consecutive checkpoints a \emph{capsule}, and the checkpoint itself a \emph{capsule boundary}. 
At a boundary, we persist enough information to continue the execution from this point when we recover from a crash.  Usually,
this means persisting the program counter along with values
in the registers (or stack frame) necessary to restart, and then atomically
setting the restart pointer to point to this information.
This approach was use by Blelloch \textit{et al.} in \cite{blelloch2018parallel} and similar to approaches by others~\cite{de2011idempotent,lucia2015simpler,colin2018termination}.   We say that an \emph{encapsulation} of a program is the placement of such boundaries in its code to partition it into capsules.

\input{cap-details}

\textbf{Capsule Correctness.}
When executing recoverable code that is encapsulated, it is possible for some instructions to be repeated. This happens if the program crashes in the middle or a capsule, or even at the very end of it before persisting the new boundary, and restarts at the previous capsule boundary. 
To be able to reason about the correctness of encapsulated programs after a crash, we define what it means for a capsule to be \emph{correct} in a concurrent setting, intuitively meaning that it can be repeated safely.

\begin{definition}
	An instruction $I$ in an execution history $E$ is said to be \emph{invisible} if $E$ remains legal even when $I$ is removed.
\end{definition}

\begin{definition}
	\label{def:caps_llcorrect}
	A capsule $\mathcal C$ inside algorithm $\alg$ is \emph{correct} if:
	\begin{enumerate}
		\item Its execution does not depend on the local values of the executing thread prior to the beginning of the capsule, and
		\item For any execution $E$ of $\alg$ in which $\mathcal C$ is restarted from its beginning at any point during $\mathcal{C}$'s execution an arbitrary number of times, there exists a set of invisible operations performed by $\mathcal C$ such that when they are removed, $\mathcal C$ appears to only have executed once in the low-level execution history.
	\end{enumerate}
\end{definition}

\begin{definition}
	A program is \emph{correctly encapsulated} if all of its capsules are correct.
\end{definition}

%% file: cap-details.tex
\textbf{Capsule Implementation.}
We now briefly discuss how capsule boundaries are implemented.

Recall that stack-allocated local variables, as well as the program
counter, are updated in volatile memory, and their new values must be
made persistent at each capsule boundary.  We therefore keep a copy of
the stack in persistent memory.  We add to each stack frame a program
counter that holds the program location of the last capsule boundary.
Furthermore each stack frame maintains two persistent copies of each
stack-allocated variable, as well as a bit indicating which of the
copies is currently valid. All these bits are kept together in a
single word as a validity mask, and this word is updated at the end of
a capsule boundary to atomically indicate which copy of each variable
is valid.  We therefore assume that each stack frame contains a
constant number of variables; in particular, that the number of such
variables is not more than the number of bits in a single word.

Heap-allocated variables, including dynamically allocated objects like
arrays, are placed immediately in persistent memory, and are handled
differently to preserve idempotence (by avoiding write-after-read
conflicts as shown in Section~\ref{sec:writecas-capsules}).
Any non-constant sized data can be allocated on the heap.

At a capsule boundary within a function, we update the values of the
stack-allocated variables that have been changed during the previous
capsule.  To do so, for each such variable, we check its validity bit,
and overwrite its \emph{outdated} copy in persistent memory. After
overwriting the outdated copies of each changed variable, we
atomically move to the next capsule by writing out a new validity mask
(bits of changed variables are flipped) along with the new program
counter.  We assume the bits and counter fit in one word (or
atomically writable object).

When making or returning from a function call we simply change the
restart pointer to point to the appropriate stack frame.  This way
when a processor crashes and restarts, it will restart from the
previous capsule boundary.

%% file: delay.tex
\section{$k$-Delay Persistent Simulations}



The performance of concurrent algorithms heavily depends on factors like contention, disjoint access parallelism, and remote accesses. While these factors are difficult to theoretically characterize, they are monumental in their effect on the algorithm's performance. Therefore, when creating algorithms for the persistent setting, it is important to be able to preserve the structure of tried and tested efficient concurrent algorithms to their persistent counterparts.

We formalize the notion of `preserving the structure of an algorithm' with the definition of a \emph{$k$-computation-delay simulation}. Intuitively, an algorithm $A$ is a $k$-computation-delay simulation of another algorithm $A'$ if, in a setting without crashes, $A$ behaves the same as $A'$, but has at most a $k$-factor slowdown per instruction.
%
%

\begin{definition}
	A concurrent program $A$ is a \emph{$k$-computation-delay simulation} of another concurrent algorithm $A'$ if $A$ follows the same steps as $A'$, but replaces each instruction $I$ of $A'$ with an implementation of $I$ that takes at most $k$ steps.
\end{definition} 

Note that the definition includes local instructions as well as accesses to base objects, and allows for several different implementations of each instruction $I$ to be used, as long as correctness is preserved. That is, as long as each step of $A'$ is replaced in $A$ with a simulation that has the same effect as the original step. 
For an algorithm $A$ to be considered a $k$-computation delay simulation of $A'$, $A$ must have at most a $k$-factor more local operations as well as shared ones. This distinction strengthens the notion of $k$-computation-delay simulations, by disallowing simulations to `trade-off' shared operations for local ones.

We now introduce some terminology that will help us discuss such simulations. We refer to $A$ as the \emph{simulation} algorithm, and $A'$ as the \emph{original} algorithm. 
For clarity, we distinguish between the \emph{base objects} or \emph{instructions} of the original algorithm and the \emph{primitive objects} of the simulation algorithm, which are the atomic objects used in the implementations of the simulated base objects. 
Each execution $E$ of the simulation algorithm $A$ \emph{maps} to the set of executions of the original algorithm $A'$ in which the accesses to $A'$'s base objects respect the partial order of accesses to these objects in $E$. 
We denote the set of executions that $E$ maps to by $E'_M$.

A stronger notion of a $k$-computation-delay simulation is one in which the amount of contention experienced by an algorithm cannot grow by more than a factor of $k$ either. Accounting for contention helps to capture the structure of an algorithm, since scalability is highly associated with keeping contention as low as possible on all accesses. 
To be able to discuss contention formally, we follow the definition of contention presented by Dwork \textit{et al.} in \cite{dwork1997contention}; the amount of contention experienced by an operation $op$ on object $O$ is the number of \emph{responses} to operations on $O$ received between the invocation and response of $op$\footnote{For atomic accesses, which don't have invocations and responses, we account for contention pessimistically; we assume that an invocation of a process's next step happens immediately after the end of its previous step.}. 
We now extend the definition of $k$-computation-delay simulations with contention taken into account as follows.

\begin{definition}
	A concurrent algorithm $A$ is a \emph{$k$-contention-delay simulation} of another concurrent algorithm $A'$ if 
	\begin{enumerate}
		\item $A$ is a $k$-computation-delay simulation of $A'$, and 
		\item for every execution $E$ of $A$, if an operation $op$ experiences $k*C$ contention in $E$, then there is an execution $E' \in E'_M$ in which the corresponding base object access experiences at least $C$ contention.
	\end{enumerate}
\end{definition} 


Persistent algorithms are tightly coupled with their recovery
mechanisms. When discussing an algorithm for a persistent setting, it
is important to also discuss how it recovers from crashes. Note that,
if all processes crash together during a system crash, simply running
a concurrent program as is in a persistent setting yields a trivial
$1$-computation-delay simulation of itself; all steps of the program
remain exactly the same. However, upon a crash, the entire program has
to be restarted, and all progress is completely lost. Thus, the
recovery time of this `simulation' is unbounded; it grows with the
length of the execution. We therefore also formalize the notion of a
\emph{recovery delay}; how long it takes for a persistent program on
any process to recover from a crash (processes can crash
independently).  Note that we consider a program to have `recovered'
when it reaches the point of the computation it was at before the
crash, and can continue the execution from there.

\begin{definition}
  A persistent program has \emph{$k$-recovery delay} if, regardless of
  the point $\rho$ of the execution at which the process on which the
  program is running crashed, the recovery on that process takes at
  most $k$ steps to arrive at $\rho$ again in a state at which the
  execution may be continued.
\end{definition}

Note that the notion of recovery delay applies to any persistent program, regardless of whether or not it is a simulation of another program. In contrast, computation delay applies to simulations, even if they are not persistent themselves.
In this paper, we consider persistent simulations that have small computation delay and small recovery delay. We say that an algorithm is \emph{X-delay free} or  if it has $c$-X-delay for a constant $c$, where $X$ is one of the types of delay we defined (computation, contention, or recovery).

%% file: cas.tex
\subsection{Recoverable Primitives}
\label{sec:dcas}
 
	One problem that arises from volatile registers and caches is that the return values of atomic operations can be lost. For example, consider a CAS operation that is applied to a shared memory location. It must atomically read the location, change it if necessary, and return whether or not it succeeded. Return values are stored in volatile registers. If a crash occurs immediately after a CAS is executed, the return value could be lost before the process can view it. When the process recovers from the crash, it has no way of knowing whether or not it has already executed its CAS. This is a dangerous situation; repeating a CAS that was already executed, or skipping it altogether, can render a concurrent program incorrect. In fact, any primitive that changes the memory suffers from the same problem.
	
	This issue was pointed out by Attiya \textit{et al.} in \cite{attiya2018nesting}. To address the problem, they present several recoverable primitives, among them a \emph{recoverable CAS algorithm}. This algorithm is an implementation of a CAS object with three operations: read, CAS, and recover. The idea of the algorithm is that when CASing in a new value, a process writes in not only the desired value, but also its own ID. Before changing the value of the object, a process must \emph{notify} the process whose ID is written on the object of the success of its CAS operation. The recovery operation checks this notification to see whether its last CAS has been successfully executed. Attiya \textit{et al.} show that their recoverable CAS algorithm satisfies \emph{nesting-safe recoverable linearizability (NRL)}, intuitively meaning that as long as recovery operations are always run immediately after crashes, the history is linearizable.
	Attiya \textit{et al.}'s algorithm uses classic CAS as a base object, and assumes that CAS operations are ABA-free, meaning that the same value is never written to the same CAS object twice.
	In particular, this disallows successful CAS operations with the same expected and new value.
	This is easy to ensure by using timestamps. 
	
	
	It turns out Attiya \textit{et al.}'s recoverable CAS algorithm satisfies strict linearizability \cite{aguilera2003strict}, a stronger correctness property than NRL. The main difference between the two properties is that while NRL only allows recovering operations that were pending when the crash happened, strict linearizability is more flexible. This means that we can define the recovery function to work even on operations that have already completed at the time of the crash. This property is very important for use in the transformations provided in the rest of the paper, in which we may not know exactly where in the execution we were when a crash occurred. 
	To satisfy strict linearizability, we need to tweak the recoverable CAS algorithm slightly, to include the use of sequence numbers on each CAS.
	In contrast to Attiya \textit{et al.}, we treat the recovery function as another operation of the recoverable CAS object, whose sequential specification is as follows.
	
\smallskip
		Each \recover{}($i$) operation $R$ returns a sequence number $seq$ and a flag $f$ with the following properties:
		\begin{itemize}
			\item If $f = 1$, then $seq$ is the sequence number of the last successful \cas{} operation with process id $i$.
			\item If $f = 0$, all successful \cas{} operations before $R$ with process id $i$ have sequence number less than $seq$.
		\end{itemize}



	
	We also further modify the recoverable CAS algorithm to create a version that has constant recovery time (instead of $O(P)$), and uses less space ($O(P)$ instead of $O(P^2)$). The pseudocode of our version of the recoverable CAS is given in Algorithm~\ref{alg:reccas}. This code also shows the tweaks that we do to the original recoverable CAS algorithm of Attiya \textit{et al.}.
	Also, in the full version of the paper~\cite{full}, we give a more detailed description of how it works, and prove it satisfies strict linearizability. Theorem~\ref{thm:reccas} summarizes the result.

	\renewcommand{\figurename}{Algorithm}
	\begin{figure}[!t!h]
	\caption{Recoverable CAS algorithm}
	\begin{lstlisting}[linewidth=\columnwidth]
class RCas {
	<Value, int, int> x;  //shared persistent
	<int, bool> A[P];     //shared persistent

	Value Read(){
		<v, *, *> = x; @\label{line:rread}@
		return v;}

	bool Cas(Value a, Value b, int seq, int i){
		<v, pid, seq'> = x;      // notify @\label{line:rcas_read}@
		if(v != a) return false;  @\label{line:rcas_check}@
		CAS(A[pid], <seq', 0>, <seq', 1>); @\label{line:rcas_cam}@
		A[i] = <seq, 0>;        // announce @\label{line:rcas_ann}@
		return CAS(x,<a, pid, seq'>,<b, i, seq>);} @\label{line:rcas_cas}@

	<int, bool> Recover(int i){
		<v, pid, seq'> = x;      // notify @\label{line:r_r1}@
		CAS(A[pid], <seq', 0>, <seq', 1>); @\label{line:r_r2}@
		return A[i];}}
	\end{lstlisting}
	\label{alg:reccas}
	\end{figure}
	\renewcommand{\figurename}{Figure}

	\begin{theorem}\label{thm:reccas}
		Algorithm \ref{alg:reccas} is a strictly linearizable, contention-delay free and recovery-delay free implementation of a recoverable CAS object.
	\end{theorem}
%

The problem of recoverability also applies to atomic write operations; if a process $p$ executes a write that may be seen and possibly overwritten by other processes, it is important that $p$ never repeat its write after a crash, since this can cause an inconsistent state.
To handle shared write operations, we reduce the problem to shared CAS operations, and then use the recoverable CAS primitive presented above.
We note that often, a shared write can be replaced with a CAS with no effect on the algorithm; that is, an algorithm $A$ that uses CASs and Writes can be simulated by algorithm $A_s$ in which each write operation is implemented with a single CAS. If the CAS fails, this is treated as if the simulated write succeeded, but was overwritten before any other process saw the value.
However, in some cases, replacing a write with a CAS does not have the same effect. Intuitively, this could occur if in algorithm $A$, the write races with a CAS operation on the same location; if the write happens before the CAS, then the CAS would fail, and the value of the write would not be overwritten. 

To handle this case, we present an algorithm  that gets rid of races between CAS and write operations on the same location. At a high level, our algorithm does this by adding a level of indirection, leading the racy CAS and Write to actually access different locations. We describe the CAS-Write algorithm in detail in Section~\ref{sec:writecas}, and show that it is computation-delay free.
Once the CAS-Write algorithm is applied, all writes can be replaced with CAS operations, which can be recovered with the recoverable CAS algorithm presented above. Hence, all of our presented simulations apply not only to programs that use CAS and read operations on shared memory, but also to those that additionally use write operations.

%% file: general.tex
\section{Persisting Concurrent Programs}
\label{sec:general}

One way to ensure that a program is tolerant to crashes is to place a capsule boundary between every two instructions.
We call these \emph{Single-Instruction} capsules. Can this guarantee a correctly encapsulated program?
Even with single-instruction capsules, maintaining the correctness of the program despite crashes and restarts is not trivial.
In particular, a crash could occur after an instruction has been executed, but before we had the chance to persist the new program counter at the boundary. This would cause the program to repeat this instruction upon recovery.

Trivially, if the single instruction $I$ in a capsule $C$ does not modify persistent memory, then $I$ is invisible, and thus $C$ is correct.
But what if $I$ does modify persistent memory? Recall that we allow non-racy writes to private persistent memory, and any CASes to shared persistent memory. A private persistent write is invisible as well, since the process simply overwrites the effect of its previous operation, and no other process could have changed it in between. 
So, we only have CASes left to handle. This is where we employ the recoverable CAS operation.

We replace every CAS object in the program with a recoverable CAS.
We show that it is safe to repeat a recoverable CAS if we wrap it with a mechanism that only repeats it if the recovery operation indicates it has not been executed; any repeated CAS will become invisible to the higher level program.
When recovering from a crash, we simply call a \texttt{checkRecovery} function, that takes in a sequence number, and calls the \texttt{Recover} operation of the recoverable CAS object. The  \texttt{checkRecovery} function returns whether or not the CAS referenced by the sequence number was successful. If it was, then we do not repeat it, and instead continue on to the capsule boundary. Otherwise, the CAS is safe to repeat. Pseudocode for the \texttt{checkRecovery} function is given in Algorithm~\ref{alg:checkRecovery}.

With this mechanism to replace CAS operations, single-instruction capsules are correct. The formal proof of correctness is implied by the proof of Theorem~\ref{thm:cap1}, which we show later.

\renewcommand{\figurename}{Algorithm}
\begin{figure}[!t!h]
	\caption{Check Recoverable CAS}
	\begin{lstlisting}[]	
	bool checkRecovery(RCas X, int seq, int pid){
	<last, flag> = X.Recover(pid);
	if (last @$>$@= seq && flag == true) return true;
	else return false;	}	
	\end{lstlisting}
	\label{alg:checkRecovery}
\end{figure}
\renewcommand{\figurename}{Figure}


We now show that this transformation applied to any concurrent program $C$ is a constant-contention-delay simulation of $C$.

\begin{theorem}\label{thm:delay}
	For any concurrent algorithm $A$, if $A'$ is the program resulting from encapsulating $A$ using single-instruction capsules, then $A'$ is a $c$-contention-delay, $c'$-recovery-delay simulation of $A$, where $c$ and $c'$ are constants. 
\end{theorem}

To prove the theorem, we first show a useful general lemma, that relates the way a simulated object is implemented to the contention-delay of a simulation algorithm. This lemma is proven in the supplementary materials.

\begin{lemma}\label{lem:delay}
	Let $A$ be a $k$-computation-delay simulation of $A'$. 
	If for every two base objects $O_{1}$ and $O_{2}$, the set of primitive objects used to implement $O_{1}$ is disjoint from the set used to implement $O_{2}$ in $A$, then $A$ is a $k$-contention-delay simulation of $A'$.
\end{lemma}

\begin{proof}[Proof of Theorem~\ref{thm:delay}.]
	Since each recoverable CAS and each capsule can be used to recover in constant time, it is easy to see that $A'$ has constant recovery delay.
	We implement each base object $O$ of $A$ by calling the operations of $O$, followed by a capsule boundary. For CASes, we implement it by replacing the CAS object with a recoverable CAS object and also calling a capsule boundary. Because both the recoverable CAS algorithm and the capsule boundary take a constant number of steps, we have shown that our transformation is a $k$-delay simulation of $A$.
	Furthermore, each recoverable CAS object uses primitive objects that are unique to it, and not shared with any other object. 
	Note that while the capsule boundary does use primitive objects that are shared among other capsule boundaries, the capsule boundaries are in fact local operations, since each process uses its own space for persisting the necessary data. So capsule boundaries do not introduce any contention.
	The rest of the proof therefore follows from Lemma~\ref{lem:delay}.
	%
\end{proof}

%% file: optimizing.tex
\section{Read-CAS Capsules}\label{sec:writecas-capsules}

Although they only consist of a constant number of uncontended steps, capsule boundaries can still be expensive in practice, as they require persisting several pieces of data and use two fence instructions. Therefore, we now discuss how to reduce the number of required capsule boundaries in a program, while still maintaining correctness.
Less capsule boundaries means more instructions per capsule. Therefore, upon a crash, several instructions may need to be repeated.

In this section, we focus on programs that use CASes and reads as their mechanisms for accessing shared data. Recall that using our CAS-Write algorithm (Section~\ref{sec:writecas}), we can extend this to programs that use writes as well. We note that this covers many concurrent programs. We show that as long as there is only one CAS operation per capsule, and this operation is the first of the capsule, the program remains correctly encapsulated.

Here we must also be aware of what local operations on heap-allocated variable do. Recall that all stack-allocated variables are written on volatile memory and any changes to them are only persistent at a capsule boundary. Therefore, there is no need to worry about inconsistencies in stack-allocated variables due to program crashes; upon a crash, any changes stack-allocated variables since the last capsule boundary will be lost, and will be safely repeated upon recovery.
However, heap-allocated variables are a different story. Recall that since there may be many heap-allocated variables, we do not handle them in the same way as stack-variables. Instead, they are directly written on persistent memory. Therefore, we must make sure that repeated code does not corrupt their values by setting capsule boundaries in between instructions that may harm each other, just like we do for the shared memory instructions. 

We note that Blelloch \textit{et al.}~\cite{blelloch2018parallel} comprehensively showed how to place capsule boundaries in non-racey persistent code to ensure idempotence. Their guideline is to create capsules that avoid \emph{write-after-read} conflicts. 
In a nutshell, these conflicts occur if a variable can be read and then written to in a persistent manner in the same capsule. If a crash occurs after such a scenario, and the code repeats itself from the read instruction, then the read sees a different value than it did originally. These conflicts can be avoided if it can be guaranteed that after a new value is persisted, the program will never repeat an earlier instruction that reads it.
Therefore, in addition to the capsule boundaries dictated by instructions on shared memory as outlined above, we also place a capsule boundary between a read of a heap-allocated location in memory and the following write to that location.
%
However, note that we don't always need to add this extra boundary; if a capsule begins with a persistent write of a private variable, any number of reads and writes to the same variable may be executed in the same capsule, since a crash in this capsule will always lead restarting the capsule, therefore overwriting the value.

We call this construction a \emph{CAS-Read} capsule.
We also allow for capsules that do not modify any shared variables at all. We call such capsules \emph{Read-Only} capsules.
Intuitively, all read operations are always invisible, as long as their results are not used in a persistent manner. 
So, a capsule that has at most one recoverable CAS operation, followed by any number of shared reads, is correct.

Note that we assume that every process has a sequence number that it keeps locally, and increments once per capsule. At the capsule boundary, the incremented value of the sequence number is persisted (along with other local values, like, for example, the arguments for the next recoverable CAS operation). Therefore, all repetitions of a capsule always use the same sequence number, but different capsules have different sequence numbers to use.

We now describe in more detail how to use the recovery function of the recoverable CAS object. We assume that there is a \texttt{crashed()} function made available to each process by the system, which returns true if the current capsule has been restarted due to a crash, and false otherwise. This assumption is realistic, since in most real systems, there is a way for processes to know that they are now recovering from a crash. We use the \texttt{crashed()} function to optimize some reads of persistent memory--- if we are recovering from a crash, we read in all local values we need for this capsule from the place where the previous capsule persisted them. Otherwise, there is no need to do so, since they are still in our local memory.
We show pseudocode for the CAS-Read capsule in Algorithm~\ref{alg:general_trans}. Read-Only capsules are a subset of the code for CAS-Read capsules.



\renewcommand{\figurename}{Algorithm}
\begin{figure}[!t!h]
	\caption{CAS-Read Capsule}
	\begin{lstlisting}[]
	if (crashed()){
		*Read all vars persisted by the 
		previous capsule into local vars.*
		seq = seq+1;
		flag = checkRecovery(X, seq, pid);
		if (!flag){ //Operation 'seq' wasn't done
			c = X.Cas(exp,new,seq, pid); }
		else {
			c = 1; }	}
	else {
		seq = seq+1;
		//exp and new are from prev capsule
		c = X.Cas(exp,new,seq,pid);} 
		*Any number of Read and local operations*
		capsule_boundary(pc, <all local values>) }
	\end{lstlisting}
	\label{alg:general_trans}
\end{figure}
\renewcommand{\figurename}{Figure}

We now show that the CAS-Read capsule is correct. This fact trivially implies that Read-Only capsules are correct as well, so we do not prove their correctness separately. We wrap up this section by showing that a transformation that applies CAS-Read and Read-Only capsules remains a $c$-contention-delay simulation for constant $c$. Intuitively, removing capsule boundaries can only improve the contention delay of a simulation.

 Note that the definition of correctness is with respect to an algorithm that contains the capsule. Here, we prove the claim in full generality; we want to show that this capsule is correct in \emph{any} algorithm that could use it. For this, we argue that its repeated operations are invisible in any execution, despite possibly arbitrary concurrent operations. 
Note that this implies correctness for any context in which the capsule might be used.

The following theorem is proven in the supplementary materials.

\begin{theorem}
	\label{thm:cap1}
	If $C$ is a CAS-Read capsule, then $C$ is a correct capsule. We also require that each process increments the sequence number before calling $CAS$.
\end{theorem}

\begin{theorem}\label{thm:construction1}
	A program that uses only CAS-Read, Read-Only, and Single-Instruction capsules is correctly encapsulated, and is a contention-delay-free simulation of its underlying program.
\end{theorem}

\begin{proof}
	Since CAS-Read, Read-Only, and Single-Instruction capsules are all correct (corollary of Theorem~\ref{thm:cap1}), by definition, a program that uses only these capsules is correctly encapsulated.
	Furthermore, since by Theorem~\ref{thm:delay},  a program encapsulated with single-instruction capsules only is a constant-contention-delay simulation of its underlying program, and CAS-Read and Read-Only capsules use strictly less instructions, programs encapsulated with these capsules are also constant-contention-delay simulations.
\end{proof}



%% file: normalized.tex
\section{Normalized Data Structures}
\label{sec:normalized}
	Timnat and Petrank \cite{timnat2014practical} defined \emph{normalized data structures}. The idea is that the definition captures a large class of lock-free algorithms that all have a similar structure. This structure allows us to reason about this class of algorithms as a whole. 
	In this section, we briefly recap the definition of normalized data structures, and show optimizations that allow converting normalized data structures into persistent ones, with less persistent writes than even our general \generalSimOpt{} would require.
	We will show two optimizations; one that works for any normalized data structure, and one that is more efficient, but requires a few more (not-too-restricting) assumptions about the algorithm.

	Normalized lock free algorithms use only CAS and Read as their synchronization mechanisms. 
	At a high level, every operation of a normalized algorithm can be split into three parts. The first part, called the \emph{CAS Generator}, takes in the input of the operation, and produces a list of CASes that must be executed to make the operation to take effect. The second part, called the \emph{CAS Executor}, takes in the list of CASes from the generator, and executes them in order, until the first failure, or until it reaches the end of the list. Finally, the \emph{Wrap-Up} examines which CASes were successfully executed by the executor, and determines the return value of the operation, or indicates that the operation should be restarted.
	Interestingly, the Generator and Wrap-Up methods must be \emph{parallelizable}, intuitively meaning that they do not depend on a thread's local values, and can be executed many times without having lasting effects on the semantics of the data structure. 

\subsection{Optimization for Normalized Data Structures}

	Our \generalSimOpt{} works for all concurrent algorithms that use the required base objects, and, in particular, works for normalized data structures. However, we can exploit the additional structure of normalized algorithms to optimize the simulation.

Note that placing capsule boundaries around a parallelizable method yields a correct capsule. 
This is implied from the ability of parallelizable methods to be repeated without affecting the execution, which is exactly the condition required for capsule correctness. The formal definition of parallelizable methods is slightly different, but a proof that this definition implies capsule correctness appears in~\cite{cohen2015efficient}.
Thus, there is no need to separate the code in parallelizable methods into several capsules according to our general construction.
Furthermore, there is also no need to use \cascam{} for some of the CAS operations performed by paralleizable methods;
for normalized data structures, we can simply surround the CAS generator and the Wrap-Up methods in a capsule, and do not need to alter them in any other way.

All that remains now is to discuss the CAS executor, which simply takes in a list of CASes to do, and executes them one by one. No other operations are done in between them. 
Note that we can convert CAS operations to use the \cascam{} algorithm, and then many consecutive CASes could be executed in the same capsule, as long as they access different objects. In the case of normalized data structures, however, we do not have the guarantee that the CASes all access different base objects. Therefore, we cannot just plug in that capsule construction as is. However, we note another quality of the CAS executor that we can use to our advantage: the executor stops after the first CAS in its list that fails.
Translated to the language of persistent algorithm, this means that we do not actually need to remember the return values of each CAS in the list separately; we only need to know the index of the last successful CAS in the list. 
Fortunately, the recovery operation of the \cascam{} algorithm actually gives us exactly that; it provides the sequence number of the last CAS operation that succeeded. Therefore, as long as we increment the sequence number by exactly $1$ between each CAS call in the executor, then after a crash, we can use the recovery function to know exactly where we left off. We can then continue execution from the next CAS in the list. Note that if the next CAS in the list actually was executed to completion but failed before the crash, there is no harm in repeating it. We simply execute it again, see that it failed, and skip to the end of the executor method. 


For a \cascam{} to work correctly, all CASes to that object must be done using the \cascam{} algorithm. Whenever a generator or wrap-up method performs CAS on an object that could also be modified by a CAS-executor, it must use \cascam{} instead of regular CAS. This is because even though the generator or wrap-up method never needs to be able to recover a CAS's results, it is still important to notify other processes of the success or failure of their last CAS.

We now discuss a method that allows removing the capsule boundary between the executor and the wrap-up. We argue that as long as we can recover the arguments and results of each executor CAS, it is safe restart the execution from the beginning of the executor.
Suppose a combined executor plus wrap-up section faults and repeats multiple times, we first argue that as long as the wrap-up part cannot overwrite the notification of any CAS in the cas-list, then in every repetition, the executor returns exactly the same index in the cas-list. Recall that we assume CAS operations are ABA-free in the original program (i.e. the object cannot take on a previous value) and that each process calls CAS using a value it previously read as the expected value. This means that if a CAS operation fails the first time, then the same CAS operation will also fail the second time. Furthermore, since we do not overwrite the result of the executor CASes outside the executor, it can always use the recovery properly to know which CAS in the list was the last to succeed.
Therefore the index returned by the CAS-executor will be the same across all repetitions. This means that the executor plus wrap-up capsule basically behaves as if there were a capsule boundary between the two methods. Since the wrap-up method is paralleizable, we know this capsule is correct.

So, to remove the capsule immediately after the executor, we need to ensure that the wrap-up does not corrupt the ability of the recoverable CAS to tell whether the most recent \emph{executor CAS} on each object succeeded, rather than just a CAS that was done in the wrap-up.
If the wrap-up does not access any CAS location accessed by the executor, this property is guaranteed.
However, if the is a CAS in the wrap-up part that accesses the same location as some CAS in the executor, we can still ensure that we can recover.
Let $C_w$ be such a CAS in the wrap-up executed by process $p$. Note that $C_w$ never needs to use the recovery function for itself; since the wrap-up is parallelizable, it is always safe to repeat $C_w$ after a crash. Therefore, when $C_w$ is executed using a recoverable CAS, it can leave out its own ID and sequence number, so that other processes do not notify $p$. Thus, the previous notification that $p$ received (i.e. a notification about $p$'s executor CAS on the same object) remains intact. 

Notice that if we have two parallelizable methods, $A$ and $B$, next to each other, we can actually put them in a single capsule as long as the inputs to $A$ and $B$ are the same whenever the capsule restarts. Since $A$ and $B$ have the same inputs, we know by parallelizability that $A$ and $B$ each appear to execute once regardless of how many times the capsule restarts. Also $A$ must appear to finish before $B$ because there was a completed execution of $A$ before any invocation of $B$. Therefore, this capsule appears to have executed only once.

So, we can avoid an additonal capsule boundary between the current iteration's wrap up method and the next iteration's generator method, as long as we now use the same notification trick in the CASes of the generator as well. So as long as there are capsule boundaries before and after each call to a normalized operation, we only need one capsule boundaries in each iteration of the main loop: only before the executor. We call this simulation the \emph{\normalOne}. The details of our encapsulation are shown in Algorithm~\ref{alg:norm_trans}. The results of this section are summarized in Theorem~\ref{thm:normalOne}.

\begin{theorem}\label{thm:normalOne}
	Any normalized data structure $N$ can be simulated in a persistent manner with constant-contention-delay using one capsule boundary in each operation.
\end{theorem}

\renewcommand{\figurename}{Algorithm}
	\begin{figure}[!t!h]
	\caption{Persistent Normalized Simulator}
	\begin{lstlisting}[]
result_type NormalizedOperation(arg_type input) 
{
	do {
		cas-list = @\textbf{\textit{CAS-Generator}}@(input);
		capsule_boundary(pc, cas-list);
		if (crashed()) {
			cas-list = read(CAS-list);
			seq = read(seq); }
		idx = @\textbf{\textit{CAS-Executor}}@(cas-list, seq);
		<output, repeat> = @\textbf{\textit{Wrap-Up}}@(cas-list, idx);		
	} while(repeat == true)
	return output;
}

int @\textbf{\textit{CAS-Executor}}@(list CASes, int seq) {
	//CASes is list of tuples <obj, exp, new>
	bool crashed = crashed();
	bool done = false;
	for (i = 0; i @$<$@ CASes.size(); i++) {
		if (crashed) {
			done = checkRecovery(CASes[i].obj, seq, p);	}
		if (!done) {
			if (!RCAS(CASes[i])) return i;	}
		seq++;	}	
	return CASes.size();
}
	\end{lstlisting}
	\label{alg:norm_trans}
	\end{figure}
	\renewcommand{\figurename}{Figure}

%% file: rwcas.tex
\section{Handling Write-CAS races}\label{sec:writecas}


Recall from Section \ref{sec:dcas} that it is often possible to replace shared variable writes with a read followed by a CAS without impacting the correctness of the algorithm.
In this section, we handle the few cases where this is not possible by implementing an array of $M$ writable CAS objects with constant-computation-delay using $O(M+P^2)$ regular CAS objects. The idea is to use a level of indirection to separate out the racy writes and CASes to different memory locations and then replace the non-racy write operations with CAS. 

Agazadeh, Golab and Woffel ~\cite{aghazadeh2014making} presented a general technique which can be used to implement a writable CAS object with constant step complexity using $O(P^2)$ CAS objects. Our algorithm is based on their technique.
At a high level, their algorithm maintains an array \texttt{B} of $O(P^2)$ CAS objects and a pointer \texttt{Ptr} which stores the index of the currently active CAS object. High-level Read() and CAS() operations read \texttt{Ptr} and apply their corresponding low-level operation to \texttt{B[Ptr]}. A high-level Write(v) operation looks for a reuseable location \texttt{B[j]} and writes $v$ into \texttt{B[j]} with a CAS (this CAS is guaranteed to succeed). Then it tries to write $j$ into \texttt{Ptr} with a CAS. If it is successful, the write operation is linearized at this CAS. Otherwise, it must have been interrupted by the successful CAS of some other write operation, so this write operation can linearize immediately before that successful CAS. The algorithm for efficiently finding the reusable CAS object \texttt{B[j]} is described later on.


We extend these ideas to implement $M$ writable CAS objects by increasing the size of \texttt{B} to $M + \Theta(P^2)$ and turning \texttt{Ptr} into an array of size $M$. The value of the $j$-th simulated object will be stored in \texttt{B[Ptr[j]]}. To find reusable locations in \texttt{B}, each process maintains a set of $\Theta(P)$ locations that it owns. These sets are disjoint and do not contain any location that is pointed to by an element of \texttt{Ptr}. Each write operation picks a reusable location from the set that it owns and if it successfully updates some pointer \texttt{Ptr[j]}, then it loses ownership of that location and gains ownership of the location that was previously in \texttt{Ptr[j]}.


Next, we explain how to locate reusable locations in \texttt{B[i]}. We first describe an implementation with amortized constant step complexity, and then breifly explain how to deamortize it. Just like in Agazadeh \textit{et al.}'s algorithm, we use a variant of Hazard Pointers~\cite{michael2004hazard} to keep track of which CAS objects are in use. To compute the reusable locations, process $p_i$ scans the announcement array and makes a list $L_i$ of all the indicies in \texttt{B} it owns which were not announced. We use helping to ensure that the locations in $L_i$ are safe to reuse. When performing a Write() operation, $p_i$ allocates from its local list $L_i$ until it runs out. Then it has to compute a new list. If each process owns $2P$ locations in \texttt{B[i]}, then the new list can be computed in $O(P)$ time and this happens at most once every $P$ Write() operations. Therefore this algorithm has constant amotrized complexity. The pseudo-code can be found in Algorithm \ref{alg:rwcas} in Appendix \ref{sec:rwcas-alg}. To deamortize, each process can maintain two lists of reusable locations and whenever it allocates from one list, it performs a constant amount of work towards populating the other. 

It is also possible to modify this algorithm to support the allocation of new writable CAS objects in case the number of writable CAS objects needed is not known in advance.

\hide{ \subsection{Proof}
\label{sec:rwcas}

\input{rwcas-alg}

We will use all caps to denote high level operations and lower case for low level operations. For example \hwrite{}{}() vs \lwrite{}(). The operation \hrc{}{}() denotes either a \hread{}{}() or a \hcas{}{}() operation. Similarly for \lrc{}(). For high level operations such as \hcas{a}{b}, the superscript denotes the high level object being operated on and the subscript denotes the low level object that was affected. Note that the subscript also represents the index returned by the call to \getobj{}().

We say that a \hwrite{}{}() operation is successful if it's \lcas{}() on line \ref{line:write_cas} is successful. Otherwise, we say the \hwrite{}{}() was unsuccessful. A successful \hwrite{}{}() operation is linearized at the \lcas{} on line \ref{line:write_cas}. If an \hwrite{}{}() $W$ was unsuccessful, there must have been a successful \hwrite{}{}() linearized between lines \ref{line:write_read} and \ref{line:write_cas} of $W$. $W$ is linearized immediately before the linearization point of the first such \hwrite{}{}() operation ('first' is abitrary, linearizing before any such \hwrite{}{}() would also be fine.). 

\begin{lemma}
\label{lem:getobj}
If a \getobj{}($a$) operation returns $b$, then \var{Ptr}[$a$] = $b$ at some configuration during the operation.
\end{lemma}

\begin{proof}
	Suppose the \lcas{} on line \ref{line:getobj_cas} succeeds. Then \var{Ptr}[$a$] = $b$ on line \ref{line:getobj_read} and we are done. If the \lcas{} fails, then it must have been inturrupted by some successful \lcas{} from line \ref{line:retire_cas} of \retire{}. In order for this \lcas{} to be successful, lines \ref{line:retire_read} and \ref{line:retire_cas} of \retire{} must have occured between lines \ref{line:getobj_read} and \ref{line:getobj_cas} of \getobj{}. Therefore line \ref{line:retire_ptr_read} occurs during the \getobj{} operation and \var{Ptr}[$a$] = $b$ at this step.

\end{proof}

\hread{a}{b} operations $R$ are linearized according to two cases:
\begin{itemize}
	\item Case A: If \var{Ptr}[$a$] = $b$ when $R$ executes it's low level \lread{b}(), then $R$ is linearized at this step.
	\item Case B: Otherwise, by Lemma \ref{lem:getobj}, there exists a configuration during the \getobj{}($a$) operation where \var{Ptr}[$a$] = $b$. Since \var{Ptr}[$a$] does not equal $b$ later on at the \lread{b} operation, there exists a step that changes \var{Ptr}[$a$] from being equal to $b$ to being not equal to $b$. $R$ is linearized immediately before the last such step.
\end{itemize}

Linearization points for \hcas{a}{b} are defined analogously to \hread{a}{b}. The linearization points of successful \hwrite{}{} operations and case A \hrc{}{} operations are guaranteed to not overlap with the linearization points of all other operations. However unsuccessful \hwrite{}{} operations and case B \hrc{} operations are all linearized immediately before some successful \hwrite{}{} operation, so they could possibly overlap. If multiple operations are linearized at the same point, we break ties by linearizing any \hrc{}{} operation before any unsuccessful \hwrite{}{} operation. We linearize the overlapping \hrc{}{} operations in order of their \lrc{} operations and we can linearize the unsuccessful \hwrite{}{} operations in any order.

Let $H$ be an execution history and let $L$ be a list of all the operations in the history in linearized order. Let $W$ be a successful \hwrite{a}{b} operation for some $a$ and $b$, and let $W'$ be the next successful \hwrite{}{} operation on the object $a$. (corner case when $W$ or $W'$ does not exist)
Consider all operations on object $a$ linearized between the linearization points of $W$ and $W'$ (excluding $W$, excluding $W'$), and call this set $S$. We will show that these operations have the correct return values. It suffices to prove two facts:

\begin{itemize}
	\item (1) The operations in $S$ all operate on the same low level object $b$.
	\item (2) All the low level operations in $S \cup \{W\}$ occur on $b$ in the same order that they are linearized.
	\item (3) No operation outside of $S \cup \{W\}$ operates on $b$ in the interval between the \lwrite of $W$ and the \lrc{}() of the last operation to be linearized in $S$.  
\end{itemize}

(1) First note that \var{Ptr}[$a$] = $b$ between the linearization points of $W$ and $W'$, so all operations in $S$ are linearized when \var{Ptr}[$a$] = $b$. In both case A and case B, \hrc{c}{d} operations are linearized at a configuration where \var{Ptr}[$c$] = $d$ for all $c$ and $d$. Therefore $S$ contains only \hrc{a}{b} operations.

(2.1) Claim: All \lrc{} operations from $S$ occur on $b$ after the \lwrite{} from $W$. 

Proof: the \lwrite{} from $W$ occurs before $W$'s linearization point and the \lrc{} operation from a \hrc{}{} operation occurs after the high level operation's linearization point in both case A and case B. All operations in $S$ are linearized after the linearization point of $W$ so the claim holds.

(2.2) Claim: \lrc{} operations within $S$ are properly ordered. 

Proof: Within case A and case B, \lrc{} operations are correctly ordered. Case B operations are linearized immediately before $W'$ so they are linearized after Case A operations. For case A operations, their low level operation happens at their linearization points and for case B operations, their low level operations happen after their linearization points. Therefore low level operations from case B happen after low level operations from case A, as required.

(3.1) No other physical writes occur in this interval.

(3.2) Claim: No case A operations outside of $S$ happens during this interval.

Proof: If a case A operation has a low level operation that happens on b in this interval, then its high level operation must have been an \hrc{a}{b} operation and it must have linearized in this interval, so it would have been part of $S$. (maybe state the property this way. If you linearize while \var{Ptr}[$a$] = $b$, then you are a \hrc{a}{b} operation)

(3.3) Claim: No case B operations outside of $S$ happens during this interval. 

Proof: follows from (3.1) because consider a case B operation $E$ outside of $S$ who's low level operation happens on $b$ in this interval. We know that this operation must have been to an object $c \neq a$ because otherwise it would be in $S$. $E$ must be linearized before the linearization point of $W$. Let $W_E$ be the last successful \hwrite{}{} operation on $c$ linearized before the linearization point of $E$. Futhermore, the linearization point of $W_E$ changes $b$ from being retired to being in use. The \lwrite{}() of $W$ happens when $b$ is retired, therefore $W_E$ is linearized before the \lwrite{}() of $W$. Consider the interval starting from the linearization point of $W_E$ to the low level operation of $E$, the \lwrite{} of $W$ happens in this interval which contradicts (3.1)

\begin{definition}
An object is said to be active if it is pointed to by an element of \var{Ptr}. The active set is the set of active elememts. For each process, we define its retired set to be the set of inactive elements that it caused to be inactive due to a CAS. These sets are clearly disjoint and every object is either active or in the retired set of some process.
\end{definition}

Lemma: no two elements of \var{Ptr} point to the same object.

\begin{lemma}
\label{lem:partition}
At every configuration, the set of objects can be disjointly partitioned into an active set and a retired set for each process. The active set is defined to be the set of objects pointed to by \var{Ptr} and the retired set for process $p_i$ is defined to be the set of inactive objects that were last retired by $p_i$. The active set will always have $m$ elements which means that all elements of \var{Ptr} point to distinct objects. 
\end{lemma}

\begin{proof}

\end{proof}

Assumption 

Assumption 1, no two elements of \var{Ptr} point to the same object at any time.

Also important to assume that the low level objects can be disjointly partitioned into $P+1$ sets, an in-use set, and a retired set for each process.

(3.1) Since \hwrite{}{} operations only write to locations in their own retired set, by Lemma \ref{lem:partition}, this claim holds up until the linearization point of $W'$. 

}

%% file: rwcas-alg.tex
\section{Pseudocode for Writable CAS Objects}
\label{sec:rwcas-alg}

	\renewcommand{\figurename}{Algorithm}
	\begin{figure}[!t!h]
	\caption{Implementing $M$ writable CAS objects using regular CAS objects. Code for proccess $p_i$}
	\begin{minipage}{.49\textwidth}
	\begin{lstlisting}
Object B[M+2*P*P];
int Ptr[M];
int free_ptr@$_\text{i}$@;

Value read(int j) {
	int idx = getObjectIdx(j);
	return B[idx].read();
}

bool CAS(int j, Value old, Value new){
	int idx = getObjectIdx(j);
	return B[idx].CAS(old, new);	
}

Value Write(int j, Value new_val) {
	int new_ptr = free_ptr@$_\text{i}$@;
	// This CAS cannot fail
	B[new_ptr].CAS(B[new_ptr], new_val); 
	int old_ptr = Ptr[j]; @\label{line:write_read}@
	if(Ptr[j].CAS(old_ptr, new_ptr)) @\label{line:write_cas}@
		free_ptr@$_\text{i}$@ = recycle(old_ptr);
}

Struct Announcement {
	int index;
	int seq;
	bool help;
} A[P];

Object getObjectIdx(int j) {
	int seq = A[i].seq+1;
	// This CAS cannot fail
	Announcement a = <j, seq, 1>;
	A[i].CAS(A[i], a); @\label{line:getobj_write}@
	int ptr = Ptr[j]; @\label{line:getobj_read}@
	A[i].CAS(a, <ptr, seq, 0>); @\label{line:getobj_cas}@
	return A[i].index;
}
\end{lstlisting}
\end{minipage}~~~
\begin{minipage}{0.49\textwidth}

\lstset{
  numbers=left,
  stepnumber=1,    
  firstnumber=40,
  numberfirstline=true
}

\begin{lstlisting}
Struct Status {
	int pid;
	bool announced;
} status[M+2*P*P];

List free_list@$_\text{i}$@;
List retired_list@$_\text{i}$@;

int recycle(int ptr) { 
	retired_list@$_\text{i}$@.push(ptr);
	// This CAS cannot fail
	status[ptr].CAS(status[ptr], <i, 0>); 
	if(free_list@$_\text{i}$@.empty()) {
		List ann_list;
		for(int j = 0; j < P; j++) {
			Announcement a = A[j];
			if(a.help) {
				int ptr = Ptr[a.index];
				A[j].CAS(a, <ptr, seq, 0>); @\label{line:retire_cas}@ }
			a = A[j];
			int idx = a.index;
			if(!a.help && status[idx] == i) {
				ann_list.push(a.index);
				// This CAS cannot fail
				status[idx].CAS(status[idx], 
				                       <i, 1>);}}
		List new_retired_list;
		for(ptr in retired_list@$_\text{i}$@) {
			if(status[ptr].announced)
				new_retired_list.push(ptr);
			else
				free_list@$_\text{i}$@.push(ptr); }
		retired_list@$_\text{i}$@ = new_retired_list;
		for(ann_ptr in ann_list)
			// This CAS cannot fail
			status[a.index].CAS(
			     status[a.index], <i, 0>); }
	return free_list@$_\text{i}$@.pop();
}
	\end{lstlisting}
	\end{minipage}~~~
	\label{alg:rwcas}
	\end{figure}
	\renewcommand{\figurename}{Figure}

\lstset{
  numbers=left,
  stepnumber=1,    
  firstnumber=1,
  numberfirstline=true
}

%% file: practical.tex
\section{Practical Concerns}
\label{sec:practical}

\textbf{Shared vs Private Model.} Recall that in the \ppm{} model, we assume that the only way for processes to communicate is through persistent memory (i.e., all shared memory is persistent). Furthermore,  the volatile memory is explicitly managed, and no automatic flushes occur. 
A similar model has been considered in several other works \cite{cohen2018inherent,friedman2018persistent,izraelevitz2016linearizability}. In this shared cache model, processes communicate through objects in volatile memory rather than persistent memory.  The values in these objects are persisted when the program issues an explicit flush instruction or when a cache line is evicted automatically. The shared model is more faithful to current cache coherent machines, while the private model helps to abstract away machine-specific flushes.

A simple transformation to convert an algorithm for the private cache variant into an algorithm that works in the shared cache variant was presented by Izraelevitz \textit{et al.} in \cite{izraelevitz2016linearizability}. This transformation simply flushes a cache line immediately after every time it is accessed (read or modified). In our experimental results in Section~\ref{sec:experiments}, we show the performance of our algorithms when the Izraelevitz transformation is applied to them. An algorithm can also be transformed from the private model to the shared model with more careful manual insertions of flushes, that may avoid some unnecessary overhead. We also show the performance of a manual transformation in our experiments. 
When transforming an algorithm from the \ppm{} model to a model in which cache lines may be automatically evicted, one needs to consider not only shared variables, but also local ones. Inconsistencies can occur in code that might repeat changes to persisted local variables (due to a crash). This can be handled again by avoiding write-after-read conflicts~\cite{blelloch2018parallel}, as is discussed for heap-allocated variables in Section~\ref{sec:writecas-capsules}.

\textbf{Compiler.} Note also that we treat shared variables differently from private ones; private heap-allocated variables may be written to many times in a single capsule, but this is disallowed for shared variables. Therefore, it is important that the compiler be able to distinguish between these two types of variables. One way to achieve this is to have the user annotate shared variables. We also need annotations to allow the compiler to determine which of our constructions should be used; we assume the user knows whether their program is in normalized form, and if so, can annotate the generator, executor, and wrap-up sections.
We also assume that for simple cases the compiler can determine if a variable is ever used again.
Therefore a capsule boundary only needs to persist the variables that may be used in the future.

\textbf{Constant Stack Frames.} Recall that for the stack-allocated variables, we assume that there is only a constant number of them (around the same as the number of bits in a word) in each stack frame. This is important to be able to atomically update the validity mask of the variables in each capsule boundary, as in Section~\ref{sec:model}.   

\textbf{CAS.} Also recall that the recoverable CAS algorithm requires storing not only the value, but also an ID and sequence number in each CAS location. This can be achieved by using a double-word CAS, which is common in modern machines.

\textbf{Flushes.} In a capsule boundary, if  all the local variables fit on
the same cache line, then we only need one fence for the capsule since the cache line gets flushed all at once in the private model. Therefore, we can also avoid having a validity bit-mask if all variables always fit in one cache line in all capsule boundaries.
On the other hand, note that on modern machines with automatic cache evictions, writing all variables on one cache line does not guarantee atomicity, since an eviction can happen part way through updating the cache line. 
However, we can still assume, as is done in~\cite{cohen2017efficient}, that writes to the same cache line are flushed in the order they are written. 
Intuitively, this is because on real machines, the following three properties generally hold: (1) total store order (TSO) is preserved, (2) individual words are written atomically, and (3) each cache line is evicted atomically.


\textbf{Memory mapping}. Note that for our algorithms to recover, we assume that after a crash, the each process can always find the memory in which the capsule boundary stored information. This requires persisting the page table. We assume that this is done by the operating system. We further assume that each process is assigned the same virtual address space as it was before the crash. These two assumptions together ensure that each process receives the same \emph{physical} address space before and after a crash. More details about the virtual to physical mapping for persistent memory is given in~\cite{chakrabarti2014atlas}.

%% file: experiments.tex
\section{Experiments}
\label{sec:experiments}
We measured the overhead of our general and normalized data structure transformations by applying them to the lock-free queue of Michael and Scott~\cite{michael1996simple}. We also compare against Romulus~\cite{correia2018romulus}, a persistent transactional memory framework, and LogQueue~\cite{friedman2018persistent}, a hand-tuned, durable and detectable queue. 
Romulus provides durability and detectability in the shared cache model as well. We use the shared cache model for our experiments because it's closer to the machine that we test on.
But this means we need to somehow translate our detectable queues from the private cache model to the shared cache model. We consider two different ways of doing this translation: automatically by Izraelevitz \textit{et al.}'s durability transformation~\cite{izraelevitz2016linearizability} or manually by hand.

We ran our experiments on an Amazon's EC2 web service with Intel(R) Xeon(R) Platinum 8124M CPU model (8 cores, 3GHz and 25MB L3 cache), and 16GB main memory.
The operating system is Ubuntu 16.04.5 LTS.

We test the performance of the execution on a real system, as in~\cite{friedman2018persistent, chakrabarti2014atlas, correia2018romulus}, assuming that the cost of flushes on current systems will be similar to what we will see in real NVM systems. 
This set up mimics the performance of battery-backed DRAM systems such as Viking NVDIMM \cite{nvdimm}.

All functions were implemented in C++ and compiled using the g++ 5.4.0 with -O3. We only measured the performance on 1-8 threads, as queues are not a scalable data structure. As in previous work~\cite{michael1996simple, friedman2018persistent}, we evaluated the performance with threads that run enqueue-dequeue pairs concurrently. In all the experiments we present, the queue is initiated with 1M nodes; however, we also tested on a nearly-empty queue and verified that the same trends occur. The \emph{flush} operation consists of two instructions: \emph{clflushopt}, and \emph{sfence}. \emph{Clflushopt} has store semantics as far as memory consistency is concerned. It guarantees that previous stores will not be executed after the execution of the \emph{clflushopt}. According to Intel, flushing with \emph{clflushopt} is faster than executing flushes using \emph{clflush}~\cite{IntelManual}. The \emph{sfence} instruction guarantees that the \emph{clflushopt} instruction is globally visible before any following store instruction in program order becomes globally visible. We omitted some fences when the ordering of the flushes is not important.
All the presented results use the \emph{recoverable CAS algorithm} that was proposed by Attiya \textit{et al.} \cite{attiya2018nesting}. In our experiments, their algorithm performed slightly better than ours and thus was the one that was presented.

Each of our tests were run for 5 seconds and we report the average throughput over 10 runs. 
In general, queues that contain less flushes perform better, which is consistent with what we expected.

In our experiments, our goal is to understand the overhead general programs would observe if they were made persistent using various methods. When we run the queue experiments, we keep in mind that these queues should be used within general programs. So, before calling each of the queue operations, the general program has to execute a capsule boundary. This is true for \emph{all} queues that we test, including the LogQueue and Romulus. Therefore, since this additional overhead would be the same for all queues tested, we omit it in our comparative experiments.
However, we note that the LogQueue and Romulus produce stand-alone data structures, that maintain more information than our queues do if the initial capsule boundary is removed. This means that in some specific contexts, for example when a few queue operations are executed consecutively, a capsule boundary can be avoided before calling LogQueue or Romulus operations, whereas our constructions still require it.
It is possible to store some extra information in our queue constructions to match the properties of the other queues, but this requires some careful manipulations, which are outside of the scope of this paper.

\paragraph{Using the Izraelevitz Construction}
One way to automatically achieve durable linearizability is to use the construction presented by Izraelevitz \textit{et al.}~\cite{izraelevitz2016linearizability}. This construction simply adds a flush after every shared memory operation.
It is a general way to make any algorithm in the private model work in the shared cache model, and requires no understanding of the semantics of the program. Figure~\ref{fig:ReadWriteFlush} shows the result of applying our transformations along with Izraelevitz's construction to the Michael-Scott queue (MSQ)~\cite{michael1996simple}.
To isolate the overhead of our transformations, we also show the performance of a Michael Scott queue with just the Izraelevitz construction. We call this the {\emph{Izraelevitz} queue and it is an upper bound on how well our transformations can perform. The result of the \generalSimOpt{} is called \emph{General}. \emph{Normalized} represents the normalized data structure transformation introduced in Section \ref{sec:normalized}.

As the general construction contains more \emph{capsule boundaries} than the \emph{Normalized} queue, we can see that \emph{Normalized} performs 1.5x better when there are 2 running threads, and by 1.15x when there are 8 threads.
Without any of the detectability transformations, the \emph{Izraelevitz} queue performs better than \emph{Normalized} by 1.3x when running 2 threads and by 1.34x when there are 8 threads.
%

\paragraph{Competitors}
Another way to make our transformed queues durable is to add flushes manually.
The flushes we add are very similar to those in Friedman \textit{et al.}'s Durable Queue \cite{friedman2018persistent}. 
The difference is that we flush both the head and tail to allow for faster recovery and we omit the return value array because it is not needed for durability. Friedman \textit{et al.} used the return value array to recover return values following a crash, but this functionality is handled by our transformations.

In these experiments, we also include optimized versions of the queues transformed with our simulations. We optimize what is stored in each capsule using insight specific to the queue algorithm.
Our goal in this is to demonstrate two separate traits of our methodology: (1) that our simulators are easy to apply automatically and without special understanding of the program, but also that (2) with our methodology, the number of local variables used really affects performance. Often, when writing code, programmers don't pay attention to how many local variables are used, and may use them superfluously. We demonstrate optimized versions in which we eliminate unnecessary local variables before applying our simulators, which yields significant performance benefits.
The \emph{General} and \emph{Normalized} implementations correspond to the automatic constructions presented in the paper in \generalSimOpt{} and \normalOne{} respectively, and the \emph{General-Opt} and \emph{Normalized-Opt} implementations correspond to their hand-optimized versions.

Our optimizations include exploiting the property where writes to the same cache line are flushed in the order they are written which presented in the Section \ref{sec:practical}, reducing unnecessary local variables and removing fences that are followed by a CAS, as it already contains a fence. We see that it implies a great difference in the \emph{Normalized} implementations, where we were able to reduce one flush. When running 1-4 threads, the \emph{Normalized-Opt} performs up to 1.57x better than the \emph{Normalized} version. In the 
\emph{General} queue, those differences are less dominant, where the \emph{General-Opt} performs up to 1.28x better than the \emph{General} queue.

We compared the manual flush version of our transformed queues with the \emph{LogQueue}~\cite{friedman2018persistent} as well as a queue written using \emph{Romulus}~\cite{correia2018romulus}. We chose the \emph{RomulusLR} version as it performed better for every thread count. Both these queues provide detectability and durable linearizability. The results are depicted in Figure~\ref{fig:CompetitorFlush2}. Our Normalized simulation performs better than Romulus which is expected because Romulus incurs extra overhead by implementing a persistent memory allocator and general transactional memory.
\emph{Romulus} outperforms our \emph{General} queue for executions with more than 5 threads. This is probably because \emph{RomulusLR} uses flat combining, so multiple update transactions are aggregated and processed with a single lock acquisition and release. 

The LogQueue is an algorithm specifically designed to make the MS queue durable and detectable. We compare it to our general methods, as well as to the specific implementation optimizations that can be applied on them in the case of the queue.
%
We found that \emph{Normalized-Opt} performs better by 1.42x on one running thread and \emph{Log} queue is better by up to 1.19x on 3-8 threads. We believe this is because \emph{Normalized-Opt} performs less overall fences compared to \emph{Log} queue, however, in some places, \emph{Normalized-Opt} performs more work in between a read and its corresponding CAS. These instructions bottleneck performance at higher thread counts. With clever cache line usage, it is also possible to reduce \emph{LogQueue} enqueues by one flush, but we did not implement this in our experiments. 
Given that the \emph{LogQueue} was specifically designed for the queue, we were impressed that our \emph{Normalized-Opt} automatic construction gets comparable performance.
We would also like to point out that the recovery function for \emph{LogQueue} requires traversing the entire queue, which can be costly for reasonably sized queues. On the other hand, our recovery function just involves loading the previous capsule and performing the recovery function of a recoverable CAS object. Since we are using Attiya et al's implementation of recoverable CAS \cite{attiya2018nesting}, recovery time is linear in the number of threads.


Figure \ref{fig:CompetitorFlushMSQ} shows how these persistent queues compare to the original MSQ without durability or detectability. From the graph, it looks like the cost paid by our transformations to ensure generality and quick recovery is not so much compared to the inevitable cost of persistence.





%% file: figures.tex
\newpage

\begin{figure}[h!]
	\centering
	\includegraphics[width=70mm,scale=0.5]{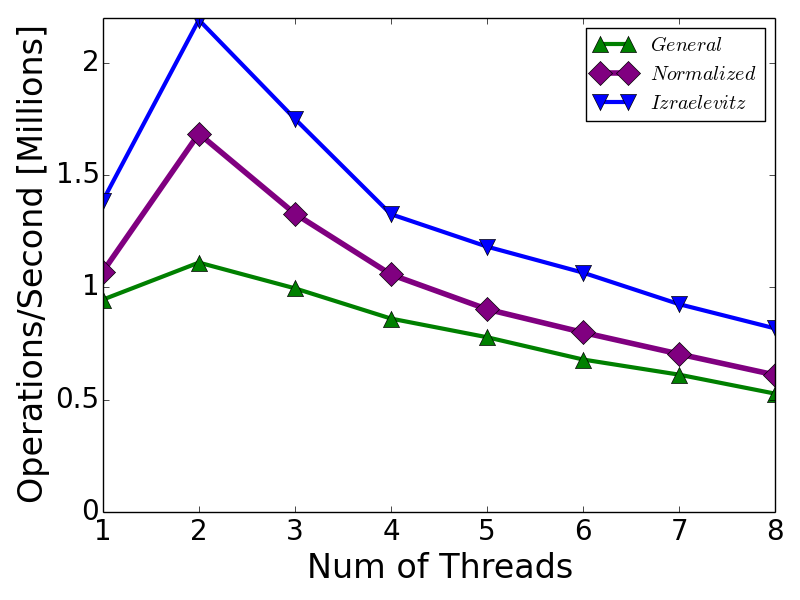}
	\caption{Throughput of transformed queues with the Izraelevitz Construction.} \label{fig:ReadWriteFlush}
\end{figure}

\begin{figure}[h!]
	\centering
	\includegraphics[width=70mm,scale=0.5]{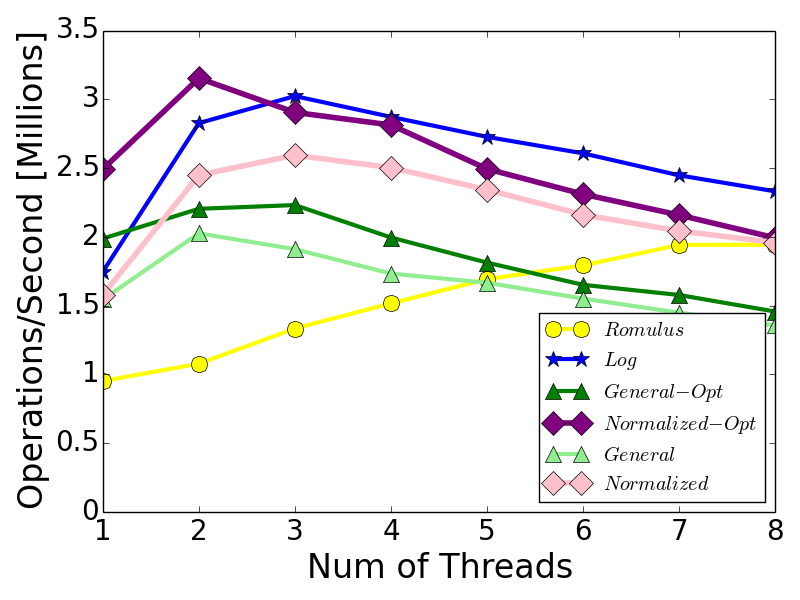}
	\caption{Comparing our transformed queues with manual flushes to prior work.} \label{fig:CompetitorFlush2}
\end{figure}

\begin{figure}[H]
	\centering
	\includegraphics[width=70mm,scale=0.5]{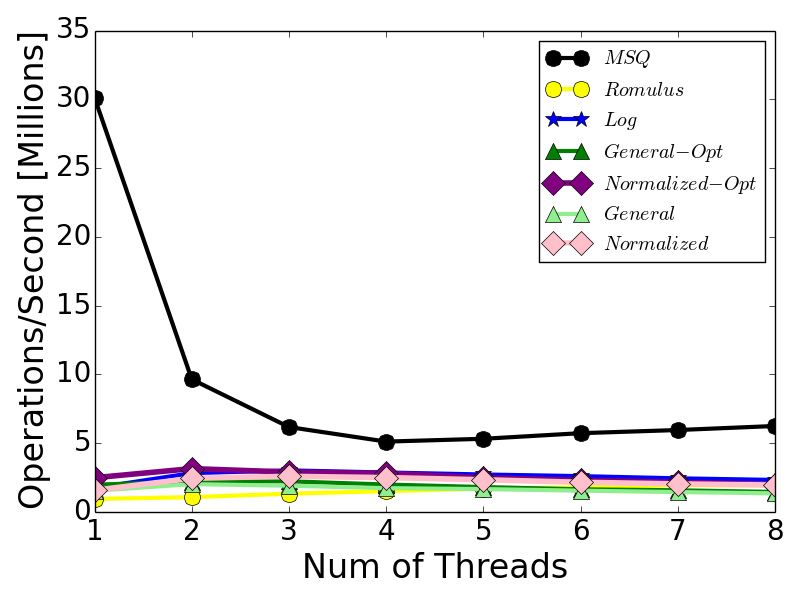}
	\caption{Comparing persistent queues to original Michael Scott Queue.} \label{fig:CompetitorFlushMSQ}
\end{figure}

%% file: appendix.tex
\appendix

\section{Recoverable CAS Details and Proof}\label{sec:reccasProof}

	The idea is simple: whenever any process executes a CAS on the object, it not only writes in its value, but also its id and the sequence number of its current operation.
However, before doing that, process $p_i$ takes a few set-up steps.
Before executing its CAS on the object, it first reads the object's state. The state is always of the form $\langle val, j, seq \rangle$, containing the process id $j$ and sequence number $seq$ of the most recent successful CAS in addition to its value. Once $p_i$ reads this state, it must notify $p_j$ of its recent success, by trying to flip the success flag in $A[j]$ from $0$ to $1$. However, $p_i$ will only change $A[j]$ if the sequence number written there is the same one that is read. In this way, a notifying process can never overwrite a more recent value in the announcement array.
Process $p_i$ can then start its own CAS. 
To do so, $p_i$ first prepares its announcement slot by writing the sequence number of this new operation in $A[i]$, with a flag set to $0$ to indicate that this CAS has not yet been successfully executed. It then proceeds to execute its CAS on the object.
To recover from a crash, $p_i$ simply needs to check the object's state to see if its value is written there, and then read its own slot in the announcement array. It is guaranteed to have been notified of its most recent success.
Algorithm \ref{alg:reccas} shows the pseudocode for this algorithm. Note that it implements a recoverable CAS object using $O(1)$ steps for all three operations. 

	We now prove that Algorithm~\ref{alg:reccas} is correct. Our correctness condition is \emph{strictly linearizability} \cite{aguilera2003strict}; that is, that all operations are linearizable, and are either linearized before a crash event, or not at all. This condition lets us safely repeat an operation if the recovery says that it didn't happen. 

When a process reads $x$ and performs a CAS on $A[i]$ for some $i$, we can view this as a notify operation.  In Algorithm \ref{alg:reccas}, lines \ref{line:r_r1} and \ref{line:r_r2} of \recover{}, as well as lines \ref{line:rcas_read} and \ref{line:rcas_cam} of \cas{} form notify operations. The purpose of the notify operation is to let a process know that its \cas{} has been successful before overwriting the value. 
The following lemma captures the key property that we require from notify operations.

\begin{lemma}
	\label{lem:notify}
	Let $N$ be an instance of a \notify{} method that reads $x=\langle *, seq, i \rangle$.
	Then after $N$'s execution, $A[i] = \langle seq, 1 \rangle$ or $A[i] = \langle seq', * \rangle$, where $seq' > seq$.
\end{lemma}

\begin{proof}
	We first show that at the first step of $N$, the sequence number in $A[i]$ is at least $seq$. This is because $x = \langle *, seq, i \rangle$. 
\end{proof}

Now we are ready to prove that Algorithm \ref{alg:reccas} is strictly linearizable. Its linearization points are also given by the following lemma.

\begin{lemma}
	\label{lem:rcas_line} 
	Algorithm \ref{alg:reccas} is a strictly linearizable implementation of a recoverable CAS object with the following linearization points:
	\begin{itemize}
		\item Each \cas{} operation that sees $v \neq a$ on line \ref{line:rcas_check} is linearized when it performs line \ref{line:rcas_read}. Otherwise, it is linearized when it performs line \ref{line:rcas_cas}.
		\item Each \readop{} operation is linearized when it performs line \ref{line:rread}.
		\item Each \recover{} operation is linearized when it returns.
	\end{itemize}
\end{lemma}

\begin{proof}
	From the linearization points of the algorithm, we see that if an operation stalls indefinitely before reaching its linearization point, then it will never be linearized. Therefore, proving linearizability is equivalent to proving strict linearizability.
	
	To show that \cas{} and \readop{} operations linearize correctly, we can ignore operations on $A$ because they do not affect the return values of these operations. At every configuration $C$, the variable $x$ stores the value written by the last successful \cas{} operation linearized before $C$. This is because the value of $x$ can only be changed by the linearization point of a \cas{} operation. So $x$ always stores the current value of the persistent CAS object.
	
	Each \readop{} operation $R$ is correct because $R$ reads $x$ at its linearization point and returns the value that was read. Each \cas{} operation $C$ is either linearized on line \ref{line:rcas_read} or line \ref{line:rcas_cas}. If $C$ is linearized on line \ref{line:rcas_read}, then it behaves correctly because $x$ does not contain the expected value at the linearization point of $C$. Suppose $a$ is the value $C$ expects and $b$ is the value it wants to write. If $C$ is linearized on line \ref{line:rcas_cas}, then we know that $x = \langle a, j, s' \rangle$ at line \ref{line:rcas_read} of $C$ for some process id $j$ and sequence number $s'$. If $C$ is successful, then $x$ contained the expected value at the linearization point of $C$ which matches the sequential specifications. Otherwise, we know that $x$ changed between lines \ref{line:rcas_read} and \ref{line:rcas_cas} of $C$. Since this recoverable CAS object can only be used in a ABA free manner, we know that $x$ does not store the value $a$ at the linearization point of $C$, so $C$ is correct to return false and leave the value in $x$ unchanged.
	
	In the remainder of this proof, we argue that \recover{} operations are correct. Let $R$ be a call to \recover{} by process $p_i$ and let $C$ be the last successful \cas{} operation by process $p_i$ linearized before the end of $R$. Let $seq(C)$ be a function that takes a \cas{} operation and returns its sequence number. We just need to show that $R$ either returns $\langle seq(C), 1 \rangle$ or it returns $\langle s', 0 \rangle$ for some sequence number $s'$ greater than $seq(C)$.
	
	First note that the sequence number in $A[i]$ is always increasing. This is because its sequence number can only change on line \ref{line:rcas_ann} of \cas{} and each process calls \cas{} with non-decreasing sequence numbers. Thus, the sequence number returned by $R$ is at least $seq(C)$ since $A[i]$ is set to $\langle seq(C), 0\rangle $ on the line before the linearization point of $C$. First we show that $R$ cannot return $\langle s', 1\rangle $ for any $s' > seq(C)$ and then we show that $R$ cannot return $\langle seq(C), 0\rangle $.
	
	Before that, we introduce the notion of a \notify{} method. The purpose of the notify operation is to let a process know that its \cas{} has been successful before overwriting the value. In our algorithm, lines \ref{line:r_r1} and \ref{line:r_r2} of \recover{}, and lines \ref{line:rcas_read} and \ref{line:rcas_cam} of \cas{} can be viewed as a \notify{} method. Each \notify{} method performs 2 steps, it reads $\langle *, s, i\rangle $ from $x$ and performs $CAS(A[i], \langle s,0\rangle , \langle s, 1\rangle )$.
	
	Now we show that $R$ cannot return $\langle s', 1\rangle $ for any $s' > seq(C)$. $R$ returns the value of $A[i]$, so suppose for contradiction that $A[i] = \langle s', 1\rangle $ at some configuration before the end of $R$. Then there must have been a \notify{} operation that set $A[i]$ to this value. This notify operation must have seen $x = \langle *, s', i\rangle $ when it performed its first step. This means that a successful \cas{} by $p_i$ with sequence number $s' > seq(C)$ has been linearized which contradicts our choice of $C$.
	
	
	Now we just need to show that $R$ cannot return $\langle seq(C), 0\rangle $, but first we prove a useful claim. If a \notify{} operation $N$ read $x=\langle *, s, i \rangle$, then after $N$ completes, $A[i]$ will never be equal to $\langle s, 0\rangle$. We know that $A[i]$ cannot be equal to $\langle s, 0\rangle$ immediately after the second step of $N$ due to the CAS it performs. Next, we show that $A[i]$ has sequence number at least $s$ after the second step of $N$. This is because in order for $x$ to be $\langle *, s, i \rangle$, there must have been a successful \cas{} operation with sequence number $s$ and process id $i$. $A[i]$ is set to $\langle s, 0 \rangle$ before the linearization point of this \cas{} operation and the sequence number in $A[i]$ is always increasing, so $A[i]$ has sequence number at least $s$ after the second step of $N$. Summarizing, we've shown that either $A[i] = \langle s, 1\rangle $ immediately after the second step of $N$ or the sequence number in $A[i]$ is larger than $s$. To finish proving the claim, all we need to show is that $A[i]$ cannot change from $\langle s, 1\rangle $ to $\langle s, 0\rangle $ after the second step of $N$. This is because $A[i] = \langle s, 1\rangle $ only after a successful \cas{} by $p_i$ with sequence number $s$ has been linearized and we have a guarantee from the user that $p_i$ will not reuse the sequence number $s$ for future \cas{} operations. 
	
	Using the previous claim, we can complete the proof by showing that there is a \notify{} operation that reads $x=\langle *, seq(C), i \rangle$ and completes before $R$ reads $A[i]$. To show this we just need to consider two cases: either $x$ is changed between the linearization point of $C$ and line \ref{line:r_r1} of $R$, or it is not. In the second case, the notify operation performed by $R$ on lines \ref{line:r_r1} and \ref{line:r_r2} reads $x=\langle *, seq(C), i \rangle$. This is because $R$ and $C$ are performed by the same process so $R$ starts after $C$ ends. In the first case, some other successful CAS operation must have been linearized after the linearization point of $C$ and before $R$ reads $A[i]$. Let $C'$ be the first such CAS operation. Then we know that $x=\langle *, seq(C), i \rangle$ between the linearization points of $C$ and $C'$. Furthermore, in order for $C'$ to have been successful, the read on line \ref{line:rcas_read} must have occurred after the linearization point of $C$ (otherwise $C'$ would not see the most recent process id and sequence number). Therefore the notify operation on lines \ref{line:rcas_read} and \ref{line:rcas_cam} of $C'$ see that $x=\langle *, seq(C), i \rangle$ and this notify completes before $R$ reads $A[i]$ as required.
\end{proof}

Lemma~\ref{lem:rcas_line}, plus the fact that each operation only performs a constant number of steps, immediately lead to the following theorem.

\begin{theorem}
\label{thm:rcas}
Algorithm \ref{alg:reccas} is a strictly linearizable, contention-delay-free and recovery-delay-free implementation of a recoverable CAS object.
\end{theorem}

\section{Proof of Lemma \ref{lem:delay}}

\begin{lemma}
	Let $A$ be a $k$-delay simulation of $A'$. 
	If for every two base objects $O_{1}$ and $O_{2}$, the set of primitive objects used to implement $O_{1}$ is disjoint from the set used to implement $O_{2}$ in $A$, then $A$ is a $k$-contention-delay simulation of $A'$.
\end{lemma}

\begin{proof}
	Consider an execution $E$ of $A$ in which an operation $op$ by process $p$ on object $O_1$ experiences $k*C$ contention. Since $O_1$ is implemented with primitive objects that are not shared with any other object, all contention experienced by $op$ must be from other operations that are accessing $O_1$. Note that each step by another process accessing $O_1$ can cause at most one contention point for $op$. Thus, there must be at least $k*C$ steps by other processes on the primitive objects $op$ is accessing within $op$'s interval. Since $A$ is a $k$-delay simulation of $A'$, and $O_1$'s primitive objects are not shared with any other base object, there must be at least $C$ other accesses of $O_1$ that are concurrent with $op$. 
	So, $E$ must map to an execution $E'$ of $A'$ in which all $C$ of these accesses to $O_1$ happen before $op$'s corresponding access, but after the last operation of $p$. Therefore, in $E'$, $op$ experiences at least $C$ contention.
\end{proof}

\section{Read-CAS Capsule Correctness Proof}

\begin{theorem}
	If $C$ is a CAS-Read capsule, then $C$ is a correct capsule. We also require that each process increments the sequence number before calling $CAS$.
\end{theorem}

\begin{proof}
	Consider an execution of $C$ in which the capsule was restarted $k$ times due to crashes.
	
	
	First note that if a crash occurred, then the capsule never uses any of the local variables before overwriting them. Therefore, its execution does not depend on local values from previous capsules. This includes the sequence number for the capsule, which must have been written in persistent memory before the capsule started, and is therefore the same in all repetitions of the capsule.
	
	Note that in all but the first (partial) run of the capsule, the \texttt{crashed} function must return true.
	Furthermore, note that the code only repeats \texttt{X.Cas()} if \texttt{checkRecovery} returns false. Due to the correctness of the recovery protocol, this happens only if each earlier operation with this capsule's sequence number has not been executed in a visible way. This means they are either linearized and invisible, or they have not been linearized at all. 
	The partial executions have not been linearized cannot become linearized at any later configuration because $X$ is strictly linearizable.
	Therefore the \texttt{X.Cas()} call is only ever repeated if the previous calls that this capsule made to it were invisible, so all but the last instance of the \texttt{X.Cas()} operations executed are invisible. Furthermore, since \texttt{X.Cas()} is a strictly linearizable implementation of CAS, the effect of instances together is that of a single CAS. 
	Note that the rest of the capsule is composed of only invisible operations; the recovery of an object is always invisible, as are Reads and local computations.
\end{proof}

\section{Formal Definition of Detectability}

\textbf{The Failure-Recovery Model.}
	An \emph{execution}, $E$, in the failure-recovery model involves three kinds of \emph{events} for each process $p$ in the system; \emph{invocation} events $I_p(op, obj)$, which invoke operation $op$ on object $obj$, \emph{response} events $R_p(op, obj)$, in which object $obj$ responds to $p$'s operation, and \emph{crash} events $C_p$. Crash events are not operation- or object-specific. On a crash event, $p$ loses all of the data in its volatile cache.
	A process $p$ takes \emph{steps} in an execution, which constitute atomic accesses to base objects, and together make an \emph{implementation} of the \emph{high-level operations} represented by the invocation and response events of the execution. Throughout this paper, we refer to steps as \emph{low-level instructions}.
	
	We require an object to provide a specific interface to be considered \emph{designed for the failure-recovery model}. In particular, an object must provide a special recovery operation that can be called after a crash. 
	
	\begin{definition}
		$A$ is a \emph{failure-recovery object} if 
		\begin{enumerate}
			\item All of its operations take in a sequence number as a parameter (in addition to any number of other parameters), and
			\item It has a special \emph{recovery} operation.
		\end{enumerate}
	\end{definition}

In this section, we present and motivate several definitions that culminate in a definition of detectability and correct encapsulation.

To compare an algorithm for the classic model to its failure-recovery version, we employ \emph{equivalent executions}, which were presented in \cite{timnat2014practical}
	\begin{definition}
		Executions $E$ and $E'$ are \emph{equivalent} if the following conditions hold:
		\begin{enumerate}
			\item In both executions, all threads execute the same operations and get identical results.
			\item The order of invocation and response points of all high-level operations is the same in both executions.
		\end{enumerate}
	\end{definition}
	
Note that the definition of equivalence does allow for the two executions to have a different sequence of low-level instructions (base object calls). This flexibility is exploited by our algorithms, which may repeat some instructions upon a crash, but guarantee that the responses of high-level operations remain unaffected.

Before defining detectability, we discuss a few other important concepts. First, we restrict how a failure-recovery algorithm can be used.
Intuitively, we require that a recovery operation be invoked immediately after every crash event. Furthermore, no response can occur for an operation that was interrupted by a crash. 
We summarize these requirements in the definition of an \emph{admissible execution.}
	\begin{definition}
		$E$ is an \emph{admissible execution} if, for every process $p$ and object $O$:
			\begin{enumerate}
				\item Each invocation by $p$ is followed by either a response or a crash for $p$. That is, we require invocation-response or invocation-crash pairs. Response events cannot appear outside of an invocation-response pair.
				\item The sequence numbers used to call $O$'s non-recovery operations are non-decreasing.
				\item After a crash event of $p$ that interrupts $p$'s operation on $O$, $p$ cannot invoke any non-recovery operations on any object until it invokes a recovery operation on $O$.
			\end{enumerate}
	\end{definition}

We are now ready to discuss what it means for a failure-recovery algorithm to be correct. For this purpose, we assume that every failure-recovery algorithm tries to simulate the behavior of some concurrent algorithm in the standard model. 
We treat the recovery operation of an object as returning an indication of whether the last operation can be safely repeated.
Intuitively, for an algorithm to be considered correct, its recovery operation must always return an acceptable signal; if we treat `repeatable' operations as never having happened, and `non-repeatable' ones as having taken effect, the resulting execution should be a legal execution of the standard-model algorithm that it simulates.
More formally, we define the \emph{projection} of a failure-recovery execution on a standard shared memory execution as follows.

	\begin{definition}
	An execution $E$ in the failure-recovery model \emph{projects} onto an execution $E_S$ in the standard model if
	\begin{enumerate}
		\item We can construct a new execution $E'$ from $E$ by removing all crash events from $E$, and treating each instance $Rec$ of the recovery operation in $E$ as follows:
			\begin{enumerate}
				\item If $Rec$ indicated that an invoked operation cannot be repeated, add a matching response immediately following $Rec$, if such a response does not appear in $E$.
				\item Otherwise, remove the last invocation before $Rec$, as well as its response (if it exists).
				\item Finally, remove the invocation and response of $Rec$ from $E$.
			\end{enumerate}
		\item $E'$ is equivalent to $E_S$.
	\end{enumerate}
\end{definition}


We consider an algorithm $A$ in the failure-recovery model to be a correct simulation of an algorithm $S$ in the standard model if all of $A$'s admissible executions project onto executions of $S$. 
Note that the definition of projections is general enough to apply to many algorithms. In particular, because our correctness definition is with respect to a specific algorithm in the standard model, this allows us to `port over' definitions and properties from the standard model into the failure-recovery model. For example, if $A$ is a failure-recovery algorithm that is correct with respect to a sequentially consistent algorithm, then $A$ is sequentially consistent in the failure-recovery model. This applies to all properties that can be defined on the executions of an algorithm. 

We now show that projection of linearizable executions is local. This means that an execution that uses two base objects that are both correct with respect to linearizable standard algorithms is itself a linearizable execution in the failure-recovery model. The following theorem formalizes this notion.
\begin{theorem}
	\label{thm:compose}
	Let $E$ be an execution in the failure-recovery model and let $O$ be the set of objects in $E$. If for all $x \in O$, the sub-execution $E|_x$ projects onto a linearizable execution, then $E$ projects onto a linearizable execution.
\end{theorem}

\begin{proof}
	Let prime denote the projection operator. By following the steps for constructing a projection, we can see that $E'|_x = (E|_x)'$ (in other words $E'|_x$ is the projection of $E|_x$) for each $x \in O$. This means that each $E'|_x$ is linearizable because $(E|_x)'$ is linearizable. Therefore, if we can show that $E$ is a legal history, then the linearizability of $E'$ would follow from the locality of linearizability. A legal history is simply a history in which processes alternate between invocations and responses, beginning with an invocation. Therefore we just need to show that operations by a single process do not overlap in $E'$. This is because a crash event can only interrupt one operation per process, and after the crash, that operation will either receive a response before the invocation of the process's next operation, or the invocation of that operation will be removed. In either case, for each process, invocations are always followed by responses in $E'$.
\end{proof} 



We now define \emph{detectable implementations}, which make specific demands of their recovery operations, to allow them to be easily usable in higher-level programs. 

	\begin{definition}
		A failure-recovery algorithm $A$ is a \emph{detectable implementation} of a standard algorithm $S$ if
		in any admissible execution $E$ of $A$, the following holds:
			\begin{enumerate}
				\item Every call $Rec$ to the recovery operation by process $p$ in $E$ returns a tuple consisting of a sequence number and a value. 
				The sequence number must correspond to the most recently invoked operation by $p$.
				The value is the response of that operation, or $\bot$ if that operation did not take effect.
				\item There is an execution $E'$ of $S$ such that $E$ projects onto $E'$.
			\end{enumerate}
	\end{definition}

This definition is a stronger property than the concept of detectability defined in~\cite{friedman2018persistent}, since it requires sequence numbers, is specific about the information that the recovery operation can return, and applies to the last operation, even if that operation was not interrupted by a crash.
However, all detectable algorithms that we are aware of in the literature actually satisfy our stronger definition \cite{cohen2018inherent,friedman2018persistent} (with minor tweaks to make them use sequence numbers).
Note that if the designer of a data structure knows the context in which the data structure will be used, it is possible that the implementation could be optimized for that specific context by not making it detectable.

\input{rwcas-alg}